\documentclass[letterpaper, 10 pt, conference]{ieeeconf}  

\IEEEoverridecommandlockouts                              
\usepackage{cite}
\usepackage{amsmath,amssymb,amsfonts}
\usepackage{graphicx}
\usepackage{algorithm2e}
\usepackage{textcomp}
\usepackage{xcolor}

\DeclareMathOperator*{\argmax}{arg\,max}
\DeclareMathOperator*{\argmin}{arg\,min}
\newtheorem{theorem}{Theorem}
\newtheorem{lemma}{Lemma}
\allowdisplaybreaks[4]

\newcommand{\online}[1]{#1}

\title{\LARGE \bf
Efficient State Estimation of a Networked FlipIt Model
}

\author{Brandon Collins, Thomas Gherna, Keith Paarporn, Shouhuai Xu, and Philip N. Brown
\thanks{This work was supported in part by AFOSR grant FA9550-23-1-0171 and the DoD UC2 program.}
\thanks{All authors are with with the Department of Computer Science,
        1420 Austin Bluffs Pkwy, Colorado Springs, CO 80918, USA
        {\tt\small \{bcollin3,kpaarpor,sxu,pbrown2\}@uccs.edu}}%
}

\begin{document}

\maketitle
\thispagestyle{empty}
\pagestyle{empty}

\begin{abstract}

The Boolean Kalman Filter and associated Boolean Dynamical System Theory have been proposed to study the spread of infection on computer networks.
Such models feature a network where attacks propagate through, an intrusion detection system that provides noisy signals of the true state of the network, and the capability of the defender to clean a subset of computers at any time.
The Boolean Kalman Filter has been used to solve the optimal estimation problem, by estimating the hidden true state given the attack-defense dynamics and noisy observations.
However, this algorithm is intractable because it runs in exponential time and space with respect to the network size. 
We address this feasibility problem by proposing a mean-field estimation approach, which is inspired by the epidemic modeling literature.
Although our approach is heuristic, 
we prove that our estimator exactly matches the optimal estimator in certain non-trivial cases.
We conclude by using simulations to show both the run-time improvement and estimation accuracy  of our approach.

\end{abstract}

\section{Introduction}
Cybersecurity has been extensively studied with 
mathematical treatments 
while assuming full and accurate information
(see, e.g. \cite{alpcan2003game,XuBookChapterCD2019}). However, the StuxNet incident \cite{farwell2011stuxnet} served as a ``wake-up call'' on studying partial and noisy information, dubbed  low-information, where partial means defenders' sensors are only employed at limited places and noisy means the sensors provide inaccurate information. This can be attributed to both the inadequacy of sensor employments (mandated by cost constraints) and the sophistication of attacks, especially the new or zero-day attacks waged by nation-state actors, also known as
Advanced Persistent Threats (APTs).

For modeling APTs, one family of studies are centered around the FlipIt model \cite{van2013flipit}.
The original FlipIt model is built around three core assumptions. (i) There is a computer that an attacker and defender can seize control over it at any time.
This is often interpreted as the attacker using an APT to gain control and the defender reinstalling the OS or reimaging a virtual machine.
(ii) Neither the attacker nor the defender has any information about who currently controls the computer.
(iii) Seizing control of the computer comes at a cost, in terms of effort by the attacker and the  lost uptime to the defender.
Thus, FlipIt can be seen as a {\em low-information} model to study how often to deploy security measures against a capable attacker.

Although many extensions to FlipIt have been proposed, including multiple computers \cite{laszka2014flipthem}, deployment of honeypots \cite{tian2019honeypot}, insider threats \cite{feng2015stealthy}, and even computer networks \cite{saha2017flipnet,banik2024flipdyn}, but there are limitations to the practical application of these models.
Primairly, the {\em low-information} assumption does not reflect reality in many settings because
Intrusion Detection Systems (IDSes) are widely deployed both on computers (or hosts) and in networks to detect cyber attacks.
Although far from perfect, IDSes provide a noisy reading of the true states of computers, which is 
not 
modeled by FlipIt-style models.

Recently, the theory of Boolean Dynamical Systems \cite{imani2016maximum} has been applied to study the APT problem with both noisy IDSes and network structures.
Originally proposed 
as a discrete-time dynamical system where each node has a state $0$ or $1$ 
and the dynamics are described by arbitrary modulo 2 operations between connected nodes.
%
The network structures and noisy observations introduce the following non-trivial state estimation problem (in addition to the FlipIt-style 
{decision problem): What is the probability that a given computer is compromised (or infected) at a specific point in time given the network structure and the history of IDS signals (i.e., alerts)?
Although this estimation problem has been optimally solved by the Boolean Kalman Filter (BKF) \cite{braga2011optimal}, it incurs an exponential time and space in terms of the size of the network because it requires computation over every possible combination of network-wide computer states.

In this work, we propose a novel estimator,
dubbed the Mean Field Analysis (MFA) estimator, because its equations resemble the MFA dynamics from networked epidemic spreading models \cite{paarporn2017networked,armbruster2017elementary,armbruster2017elementary2,XuTAAS2012,WangSRDS03,mei2017dynamics,nowzari2016analysis,gleeson2012accuracy,pare2020modeling,silva2024accuracy}.
In the proposed MFA method, we keep track of the probability of infection for each node (which requires $n$ real numbers) rather than a distribution over all possible states (which requires $2^n$ real numbers), leading to a polynomial-time estimation algorithm.
Our results are summarized as follows:
\begin{enumerate}
    \item We propose a novel estimator (the MFA algorithm) based on an approximation of Bayes' Law using the Maximum Entropy Principle.
    \item We prove that all steps of the MFA algorithm are exact and efficient computations of the corresponding step of the BKF algorithm, under the condition that the prior belief is an appropriate maximum entropy distribution (Theorems~\ref{thm:XMPi},\ref{thm:XTPi}).
    \item We run simulations to compare the performance of the MFA and BKF estimators.
\end{enumerate}
Thus, we provide a novel estimator that is both theoretically justified and practical for running on large real-world networks.
Notably, our proposed method is suitable for real-time computation on large networks with hundreds of nodes, where the state-of-the-art approach BKF takes hours on graphs with just 14 nodes.
Notably, a previous approach also developed a polynomial-time algorithm to estimate the BKF based on particle filters \cite{imani2018particle}.
In contrast to their approach, our estimation approach leverages specific problem characteristics to draw parallels with MFA as it is applied in epidemic theory.

\section{Problem Formulation}
\subsection{Model}
We model the spread of cyber attacks in a computer network where the true state of each node is unknown.
This can be considered a hidden Markov model with boolean states and given noisy state observations
at each discrete time step.
We consider networks with $n$ nodes and let $\{x_t;t=0,1,2,\dots\}$ be a state process where $x_t\in\{0,1\}^n$ is the state of the process at time $t$.
We use $x_{t,l}$ to denote the state of computer $l$ at time $t$ where $x_{t,l}=1$ means the computer is compromised and $x_{t,l}=0$ means the computer is secure. 
Let $X=[x^1,x^2,\ldots, x^{2^n}]$ be an $n \times 2^n$ matrix where each $x^i\in\{0,1\}^n$ represents a unique state.
Together, all of the columns $x^i$ of $X$ include all possible states in $\{0,1\}^n$.
Occasionally, it is convenient to abuse notations somewhat and treat $X$ as a set, namely $X=\{x^1,x^2,\dots,x^{2^n}\}$.
We denote a distribution over all $2^n$ states as $\Pi\in\Delta(X)$, where $\Delta(X)=\{\Pi\in[0,1]^{2^n}\mid \sum \Pi =1 \}$ is the space of all distributions over $X$.
Further, $P_t\in [0,1]^n$ a vector of distributions of each node, where $P_t,l$ encodes the marginal belief that node $l$ is infected at time $t$.
To avoid confusion, we use consistent symbols for indexing, $t\in \{0,1,2,\ldots\}$ is a time index, $l,k,r\in V$ are node indexes, and $i,j\in\{1,2,\ldots,2^n\}$ are indexes over $X$ (e.g., $x^i,x^j\in X$).
We will use the indexes $i,j$ from $x^i,x^j$ to refer to the related probability in $\Pi$ as $\Pi_i,\Pi_j$.
Due to the amount of indexing required, we only use time indexes, when necessary, to describe the evolution of $x_t,\Pi_t$ and $P_t$ over time.

Underlying the computers is a directed graph or network structure $G=(V,E)$ where the set of computers (i.e., nodes) is denoted by $V=\{1,2,\ldots,n\}$, and the links (or connections) between them are given by directed edge set $E$.
A computer $l$ can propagate an attack to computer $k$ if $(l,k)\in E$.
The set of connections models 
the communications permitted by the cybersecurity policy (known as {\em attack-defense} structure \cite{XuTAAS2012}), rather than representing the physical network structure or topology.
Let $D_l=\{k\in V\mid (k,l)\in E\}$ be the in-neighbor set of a node $l$.

At each time step, if an in-neighbor $k$ of node $l$ is compromised, then the attack will successfully compromise node $l$ with probability $\rho_{kl}\in[0,1]$.
At each time step the defender may opt to clean each node individually, indicated by $a_t\in\{0,1\}^n$, where $a_{t,l}=1$ means the defender cleans node $l$ at time $t$ even if the defender is not certain whether node $l$ is compromised or not, and conversely for $a_{t,l}=0$.
For model generality, we assume the cleaning will have a failure probability $\alpha\in[0,1]$, meaning a cleaning will succeed with probability $(1-\alpha)$.
Finally, we assume that node $1$ models an external threat and is always compromised, namely $x_{t,1}=1$ for all $t$, and cannot be cleaned, namely $a_{t,1}=0$ for all $t$.
Particularly, any node $l$ that has an edge from node 1, namely $(1,l)\in E$, is thought to be vulnerable to external threats and any subset of nodes may have connections from node 1.
This assumption is made 
out of mathematical convenience and can be seen as equivalent to the environment (in which a network resides) such that the nodes in the network may be susceptible to pull-based (e.g., drive-by download) attacks that were first mathematically modeled in \cite{XuTAAS2012}); moreover, \cite{kazeminajafabadi2024optimal} considers an equivalent version where the attack is not explicitly modeled in the network.

The dynamics of the true network state is given by
\begin{equation}\label{eq:xt dynamics}
\begin{aligned}
    \Pr(x_{t+1,l}=1\mid x_t,a_t)&=
    (1+(\alpha-1) a_{t,l})\bigg(x_{t,l}
    +\\
    &(1-x_{t,l})\bigg[1-\prod_{k\in D_l}(1-\rho_{kl}x_{t,k})\bigg]  \bigg),
\end{aligned}
\end{equation}
which gives the probability that node $l$ is compromised given the previous state and cleaning action $(x_t,a_t)$, the graph structure, and the attack success parameters $\rho_{kl}$.
The cleaning success term $(1+(\alpha-1) a_{t,l})$ denotes the probability the node is not cleaned successfully, and the subsequent term denotes the probability that $l$ is compromised given the previous state of the network $x_t$.
Particularly, if $l$ was not previously compromised, namely $x_{t,l}=0$, then the compromise probability is determined to the term with the product, which is
the probability at least one successful attack comes from any in-neighbor.

After the attack and cleaning process has concluded, we assume that we receive a noisy measurement of the true state of the node, denoted by $y_t\in\{0,1\}^n$,
where $y_{t,l}=1$ indicates that IDS (potentially incorrectly) detected $x_{t,l}=1$ and likewise for $y_{t,l}=0$, namely that the noises are incurred by the False Positives and False Negatives of IDS.
The observation $y_t$ is noisy as follows
\begin{equation}
y_{t,l}=\begin{cases}
        1 \mbox{ if }x_{t,l}=1 & \mbox{with probability }p\\
        0 \mbox{ if }x_{t,l}=1 & \mbox{with probability }1-p\\
        0 \mbox{ if }x_{t,l}=0 & \mbox{with probability } q\\
        1 \mbox{ if }x_{t,l}=0 & \mbox{with probability }1-q\\
    \end{cases}
\end{equation}
where $p,q\in[0,1]$ are the {True Positive} and {True Negative} probabilities respectively.
Occasionally, we use $\tilde{p}=1-p$ and $\tilde{q}=1-q$ for compactness.

Given that the true state $x_t$ is unknown but dynamics \eqref{eq:xt dynamics}, cleanings $a_t$, and observations $y_t$ are known, we seek an estimator $\hat{x}_t\in \{0,1\}^n$ of
the true state $x_t$.
To measure the quality of an estimator, we use the mean-squared error (MSE) criterion 
\begin{equation}
    C(x_t,\hat{x}_t)=\mathbb{E}(||x_t-\hat{x}_t(a_{0:t-1},y_{1:t})||^2_2\mid a_{0:t-1},y_{1:t})
\end{equation}
where $||\cdot||_2$ is the $\ell_2$ norm.
This objective function counts the expected number of nodes incorrectly estimated given the history of observations and cleanings.
Thus, we seek estimators that can 
(approximately) achieve optimization with respect to MSE, namely 
\begin{equation}\label{eq:MSE optimality}
    \hat{x}^*_t\in \argmin_{\hat{x}_t\in \Psi} C(x_t,\hat{x}_t)
\end{equation}
where $\Psi$ is the set of all functions that map $a_{0:t-1},y_{1:t}$ onto $\{0,1\}^n$.


\subsection{The Boolean Kalman Filter }

\RestyleAlgo{ruled} 
\begin{algorithm}
\caption{The Boolean Kalman Filter Algorithm\label{algo:BKF}}
\KwData{$\Pi_{0\mid 0}$}
\KwResult{$\hat{x}^*_1,\hat{x}^*_2,\hat{x}^*_3,\cdots$}
\For{t=1,2,\dots}{
            \nl$\Pi_{t|t-1}\gets M_t(a_{t-1})\Pi_{t-1|t-1}$\label{eq:alg:MPi}\;
            \nl$\Pi_{t|t}\gets \frac{T_t(y_t) \circ\Pi_{t\mid t-1}}{(T_t(y_t))^\top\Pi_{t\mid t-1}}$\label{eq:alg:TPi}\;
            \nl$\hat{x}^*_t=\overline{X\Pi_{t|t}}$\label{eq:alg:XPi}\;
}
\end{algorithm}

The optimal estimator on the preceding problem, known as the Boolean Kalman Filter (BKF), has been given in \cite{kazeminajafabadi2024optimal} (with minor technical differences).
Before giving the full details of the Algorithm~\ref{algo:BKF}, we discuss the high-level ideas behind it.
The algorithm has three steps at each iteration,
where the probability distribution $\Pi_{t-1|t-1}\in\Delta(X)$ is updated to $\Pi_{t|t}\in\Delta(X)$, which is leveraged to produce the optimal MSE estimator.
The distribution $\Pi_{t-1|t-1}$ is known as the \emph{prior belief} and is the existing belief across all states $X$.
In the first step, the matrix-vector product is taken
\begin{equation}\label{eq:attack update}
    \Pi_{t|t-1}=M_t(a_{t-1})\Pi_{t-1|t-1},
\end{equation}
where $M_t$ is a $2^n\times 2^n$ column stochastic matrix that represents the expected dynamics of \eqref{eq:xt dynamics} with the cleaning vector $a_{t-1}$.
which is then operated on by observation information (encoded as vector $T_t(y_t)$) to obtain the updated belief $\Pi_{t\mid t}$, as follows:
\begin{equation}\label{eq:observation update}
    \Pi_{t\mid t}=\frac{T_t \circ\Pi_{t\mid t-1}}{(T_t)^\top\Pi_{t\mid t-1}},
\end{equation}
where $\circ$ denotes the element-wise multiplication operator and $\top$ denotes the transpose operator.
The vector $T_t(y_t)$ leverages the actual observation $y_t$ with the True Positive and True Negative probabilities $p,q$ to adjust the values of belief $\Pi_{t\mid t}$ accordingly.
Finally, the optimal MSE estimate is provided by
\begin{equation}\label{eq:BKF:xhat}
    \hat{x}_t=\overline{X\Pi_{t\mid t}},
\end{equation}
where the overline operator $\overline{\cdot}$ is the element-wise rounding operator.
The drawback of the BKF is that all three components, $\Pi_{t\mid t}$,$M_t$,$T_t$, are all exponentially sized, namely at the magnitude of
$2^n$.
This automatically implies the BKF 
algorithm is not feasible
for sufficiently large
networks.

We now provide the full details of $M_t(a_{t-1})$ and $T_t(y_t)$, beginning with $M_t(a_{t-1})$.
Given state $x^j\in X$, the probability that node $l$ is compromised in the next time step is given by 
\begin{equation}
\begin{aligned}
    \eta^j_l=&\bigg(1+a_{t-1,l}(\alpha-1)\bigg)\\
    &\quad\bigg(x^j_l+(1-x^j_l)\bigg[1-\prod_{r\in D_l}1-\rho_{rl}x^j_r\bigg]\bigg)
\end{aligned}
\end{equation}
which directly follows the dynamics \eqref{eq:xt dynamics}, while taking into account cleaning $a_{t-1}$ and the previous state of node $x^j_l$.
We leverage $\eta^j_l$ to define the column stochastic matrix $M_t$, which is the transition matrix that defines the hidden Markov chain over the $2^n$ states.
Specifically, entry $(M_t)_{ij}$ is the probability that state $x^j$ will transition to $x^i$, namely:
\begin{equation}\label{eq:def Mt}
\begin{aligned}
    (M_t)_{ij}&=\Pr(x_t=x^i\mid x_{t-1}=x^j,a_{t-1}) \\
    &=\prod^n_{l=1}\bigg(\eta^j_l x^i_{t,l}+(1-\eta^j_l)(1-x^i_{t,l}) \bigg).
\end{aligned}
\end{equation}
The matrix $M_t$ fully describes the expected underlying dynamics of the hidden Markov model, but observation $y_t$ is available to improve the estimate.
To leverage $y_t$, we construct vector $T_t\in[0,1]^{2^n}$ such that
\begin{equation}\label{eq:def Tt}
\begin{aligned}
    (T_t(y_t))_{i}&=\Pr(y_t\mid x_t=x^i)\\
    &=\prod^n_{l=1}\Pr(y_{t,l}\mid x_t=x^i)\\
    &=\prod^n_{l=1}\bigg[y_{t,l}\bigg(px^i_{l}+(1-q)(1-x^i_{l})\bigg)\\
    &\qquad+(1-y_{t,l})\bigg(q(1-x^i_l)+(1-p)x^i_l\bigg)\bigg]
\end{aligned}
\end{equation}
which is the probability that observation $y_t$ occurred given that the true state was $x^i$.


\section{Approach}
\subsection{An Approximate Approach}

To remedy the computational complexity of the BKF algorithm, we present a polynomial-time approximation based on Bayes' Law conditioned on the belief that each node is
compromised.
Particularly, we propose tracking only the belief $P_t\in [0,1]^n$ that corresponds to the probability that each node is compromised.
This corresponds to the product $X\Pi_{t\mid t}$, which has the same interpretation.
In particular, we use Bayes' Law to condition on the fact that we only have belief $P_t\in [0,1]^n$ as the prior instead of
$\Pi_{t\mid t}\in\Delta(X)$.
The dynamics can then be given by
\begin{equation}\label{eq:bayes update}
\begin{aligned}
    &P_{t+1,l}=\Pr(x_{t+1,l}=1\mid P_{t},a_t,y_{t+1})=\\
    &\begin{cases}
        \frac{\Pr(y_{t+1,l}=1\mid P_t,a_t,x_{t+1}=1)\Pr(x_{t+1,l}=1\mid P_t,a_t)}{Pr(y_{t+1,l}=1\mid P_t,a_t)} &y_{t+1,l}=1\\
        \frac{(1-\Pr(y_{t+1}=1\mid P_t,a_t,x_{t+1}=1))\Pr(x_{t+1,l}=1\mid P_t,a_t)}{(1-Pr(y_{t+1,l}=1\mid P_t,a_t))} &y_{t+1,l}=0
    \end{cases}
\end{aligned}
\end{equation}
via Bayes' Law.
It then suffices to derive each of the three terms as a function of $P_t,a_t,y_{t+1}$ to complete our derivation of the MFA estimator.
The first term of the numerator is easy to show directly, $\Pr(y_{t+1,l}=1\mid P_t,a_t,x_{t+1}=1)=p$;
then, by using this result and the law of total probability, the denominator can be solved as
\begin{equation*}\label{eq:vt+1=1|u_t=0}
\begin{aligned}
    \Pr(y_{t+1,l}&=1\mid P_t,a_{t})
    =p\Pr(x_{t+1,l}=1\mid P_t,a_t)\\
    &
    \quad+\tilde{q}(1-\Pr(x_{t+1,l}=1\mid P_t,a_t))
\end{aligned}
\end{equation*}
which is described using
the remaining term $\Pr(x_{t+1,l}=1\mid P_t,a_t)$.
This term is potentially difficult to compute
because many distributions $\Pi\in \Delta(X)$ can be consistent with $P_t\in[0,1]^n$, which we denote by set $S(P_t)=\{\Pi\in \Delta(X)\mid X\Pi=P_t\}$.
This raises two questions. First, evaluating $\Pr(x_{t+1,l}=1\mid P_t,a_t)$ likely requires finding a distribution over $S(P_t)$, rather than a single distribution $\Pi$.
It is conceivable that novel dynamics would need to be derived because such a distribution would evolve given $y_{t+1},a_{t}$ and dynamics \eqref{eq:xt dynamics}.
Second, to numerically evaluate that term would require sampling $S(P_t)$ with the desired distribution, which itself is non-trivial owing to the fact that $X\Pi_{t|t}=P_t$ defines a system of $n$ equations that have $2^n$ variables.

Regardless, it is likely any computation over the whole set of $S(P_t)$ would take place in exponential time, which is also true for
computing $X\Pi_{t|t}$, leading to the presumption that directly evaluating $\Pr(x_{t+1,l}=1\mid P_t,a_t)$ is not 
tractable.


\subsection{The Mean Field Analysis Estimator}
\begin{algorithm}
\caption{The Mean Field Analysis Algorithm\label{algo:MFA}}
\KwData{$P_0\in [0,1]^n$}
\KwResult{$\hat{x}_1,\hat{x}_2,\hat{x}_3,\cdots$}
\For{t=1,2,\dots}{
\For{$l\in V$}{
        \nl $P'_{t,l}\gets (1+(\alpha-1)a_{t,l})$\\
    $\bigg( P_{t-1,l}+(1-P_{t-1,l})\bigg[1-\prod_{k\in D_l} (1-\rho_{kl}P_{t-1,k})\bigg]\bigg)$\label{eq:alg:MFA eq}\;
        \nl $P_{t,l}\gets\frac{py_{t,l}P'_{t,l}+(1-y_{t,l})\tilde{p}P'_{t,l}}{y_{t,l}[pP'_{t,l}+\tilde{q}(1-P'_{t,l})]+(1-y_{t,l})[\tilde{p}P'_{t,l}+q(1-P'_{t,l})]}$\label{eq:alg:bayes}\;
}
\nl $\hat{x}_t\gets \overline{P_{t}}$\label{eq:alg:estimator}\;
}
\end{algorithm}

Given that we cannot directly evaluate Bayes' Law with belief $P_t$, the main premise of our estimator is to leverage the Mean Field Analysis (MFA) equation to provide a tractable equation in place of $\Pr(x_{t+1,l}=1\mid P_t,a_t)$.
The MFA equations have long been proposed for a similar purpose in epidemic modeling (see, e.g. \cite{paarporn2017networked,nowzari2016analysis}), where each node is modeled individually.
The dynamics are given 
to \eqref{eq:xt dynamics}, with the main difference that in epidemic modeling, there is no notion of cleaning $a_t$, 
but rather some natural recovery process is assumed.
However, both models lead to dynamics similar to \eqref{eq:attack update} where there exists an underlying $2^n$ state Markov chain that is difficult to analyze.
Thus, the so-called MFA equation has been proposed, given in Line~\ref{eq:alg:MFA eq} of Algorithm~\ref{algo:MFA}.
Particularly, it can be seen that the MFA equation is structurally similar to dynamics \eqref{eq:xt dynamics} with the exception that states $x_t$ have been replaced with belief $P_t$.
With respect to networked 
models,  there have been considerable numerical efforts to evaluate the performance of Mean Field estimators \cite{gleeson2012accuracy,silva2024accuracy}, and it has been shown via two prominent epidemic models \cite{armbruster2017elementary,armbruster2017elementary2} that in a special case, the variations of the MFA dynamics do converge to the true dynamics.

We now provide a high-level description of the proposed Algorithm~\ref{algo:MFA} that aims to approximate the results of
the BKF Algorithm~\ref{algo:BKF}.
Both algorithms function in 3 sub-steps (Line 1-3) at each iteration; each step of the MFA algorithm functions analogously to the BKF.
Note that the MFA algorithm requires an extra loop over $V$ as the matrix products can no longer be used to do that implicitly.
Further, it can be seen that Line~\ref{eq:alg:MFA eq} utilizes the MFA equation to estimate the expected dynamics \eqref{eq:xt dynamics}, given prior belief $P_{t-1}$.
This produces an intermediate distribution $P'_{t}$, which is then updated in Line~\ref{eq:alg:bayes} with respect to the received observation $y_t$ utilizing Bayes' law as given in \eqref{eq:bayes update}, producing the new belief $P_t$.
Finally, in Line~\ref{eq:alg:estimator}, we round the posterior distribution $P_t$ in an element-wise fashion to obtain the MFA estimator $\hat{x}_t$.

It should be noted that, at each time step, the algorithm requires an operation over each node $l\in V$, and within that loop, a loop over all in-neighbors of $l$.
Thus, the MFA algorithm can be pessimistically regarded as running in $\mathcal{O}(nD^*)$ time, where $D^*=\max_{l\in V}|D_l|$ is the maximum number of in-neighbors among all nodes.
Additionally, it can be seen that our MFA algorithm can be seen as a \emph{distributed} algorithm in the sense that, for any node $l$ to compute their next estimate $P_{t+1,l}$, they only require parameters relating to their in-neighbors $D_l$. 

\section{Results}
Although justified in the context of epidemic models, 
using the MFA equation to evaluate an approximate Bayes' law is not
well justified yet in the context of 
state estimation. To address this gap, we now prove
that all three steps of the MFA algorithm can be derived from the BKF via a special case of $\Pi\in S(P)$.
To see this, we revisit the 
attempt to compute (rather estimate) $\Pr(x_{t+1,l}=1\mid P_t,a_t)$.
As discussed previously, computing over a distribution
over
$S(P)$ is difficult and is likely to take exponential time.
Thus, one approach to estimating this probability is to select a single $\Pi\in S(P)$, but how to pick?
A disciplined answer is the distribution $\Pi\in S(P)$ that maximizes entropy.
Formally, given $P\in[0,1]^n$, we define
\begin{equation}\label{eq:Pistar def}
    \Pi^*(P)\in\argmax_{\Pi\in S(P)} \sum_{i}-\Pi_i\ln \Pi_i,
\end{equation}
where $\Pi^*(P)$ is the distribution that maximizes the entropy, constrained to be consistent with belief $P$.
To justify why the maximum entropy distribution is a disciplined choice, we appeal to the Maximum Entropy Principle.
This principle addresses a common problem in statistical estimation, where many possible underlying distributions are consistent with available information.
The authors \cite{guiasu1985principle} describe the appeal of selecting the maximum entropy distribution as \emph{``what we need is a probability distribution which ignores no possibility
subject to the relevant constraints.''}
The intuition is as follows: although it would be more desirable to have a lower entropy distribution, there is no way to select the correct lower entropy distribution consistent with the observations.

Before we present the results on deriving each step of the MFA algorithm,
we first give a theorem that derives $\Pi^*(P)$ in closed form.
\begin{theorem}\label{thm:max entropy}
    Given any belief $P \in [0,1]^n$, the distribution \\ $\Pi^*(P) \in \Delta(X)$ defined by
    \begin{equation}
        \Pi^*(P)_i=\prod_{l\in V} ( P_lx^i_l+(1-x^i_l)(1-P_l) )
    \end{equation}
    for all $x^i \in X$ is the unique distribution that solves \eqref{eq:Pistar def}.
\end{theorem}
Theorem \ref{thm:max entropy} 
states that the distribution $\Pi^*(P)$ 
that maximizes entropy while being consistent with the belief $P$ is the product distribution over $X$ induced by $P$.
Due to space constraints we defer the proof to an online version \cite{collins2025efficient}.
Using this form, we now evaluate $\Pr(x_{t+1,l}=1\mid \Pi^*(P_t),a_t)=XM_{t+1}\Pi^*(P_t)$, which is Line~\ref{eq:alg:MPi} of the BKF Algorithm~\ref{algo:BKF}.
\begin{theorem}\label{thm:XMPi}
Let $M_t$ be constructed appropriately via \eqref{eq:def Mt} for any $a_{t-1}$. Then, $\Pr(x_{t+1,l}=1\mid \Pi^*(P_t),a_t)$ is computed explicitly in $\mathcal{O}(n)$ via
    \begin{equation}\label{eq:Pt dynamics}
    \begin{aligned}
        &\Pr(x_{t+1,l}=1\mid \Pi^*(P_t),a_t)=P'_{{t+1},l}=(X M_{t+1} \Pi^*(P_t))_l\\
        &=(1+(\alpha-1)a_{t,l})\\
    &\qquad \left( P_{t,l}+(1-P_{t,l})\bigg[1-\prod_{k\in D_l} (1-\rho_{kl}P_{t,k})\bigg]\right).
    \end{aligned}
    \end{equation}
\end{theorem}
The proof is again deferred to the online version \cite{collins2025efficient}. 
%
Next, we proceed directly to Line~\ref{eq:alg:TPi}, 
showing that it may also be derived from the BKF, given $\Pi_{t\mid t-1}=\Pi^*(P'_{t})$. 
\begin{theorem}\label{thm:XTPi}
    Given a distribution $\Pi_{t\mid t-1}$ such that $\Pi_{t\mid t-1}=\Pi^*(P)$ for some $P'_t\in[0,1]^n$, then \eqref{eq:observation update} is computed in $\mathcal{O}(n)$ is computed via 
    \begin{equation}
    \begin{aligned}
        &P_{t,l}=(X\Pi_{t\mid t})_l=\bigg(X\frac{T_t\circ \Pi_{t\mid t-1}}{(T_t)^\top\Pi_{t\mid t-1}}\bigg)_l=\\
        &\frac{py_{t,l}P'_{t,l}+(1-y_{t,l})\tilde{p}P'_{t,l}}{y_{t,l}[pP'_{t,l}+\tilde{q}(1-P'_{t,l})]+(1-y_{t,l})[\tilde{p}P'_{t,l}+q(1-P'_{t,l})]}.
    \end{aligned}
    \end{equation}
\end{theorem}
Again, we defer the proof to \cite{collins2025efficient}.
Interestingly, this indicates that the Bayesian update of the observation is equivalent to the BKF's observation update in this context.

These theorems provide more formal evidence that the idea of using Bayes' law given the belief $P_t,P'_t$ and leveraging the maximum entropy distribution $\Pi^*(P)$ is indeed closely related to the actual function of the BKF algorithm.
The performance of the MFA algorithm thus depends on how well the BKF function performs when restricted to maximum entropy distributions.
The unrestricted variant is capable of computing lower entropy representative distributions, which may potentially lead to better estimates of $\hat{x}_t$.
This distinction can be regarded as the precise loss of information from storing only $n$ real numbers, in contrast to
storing the full distribution over $2^n$.

It is important to note that although the MFA approaches can follow each step of the BKF algorithm individually, it is typically true that $\Pi_{t|t-1}\neq \Pi^*(X\Pi_{t|t-1})$. This implies that even if the prior $\Pi_{t-1|t-1}=\Pi^*(P_{t-1})$ maximized entropy, the assumption of Theorem~\ref{thm:XTPi} is 
not be satisfied and often 
$P_t\neq X\Pi_{t|t}$.

\subsection{Simulation Results}
We now provide simulations to verify the effectiveness of the MFA algorithm based on the intuition of Theorems~\ref{thm:XMPi} and \ref{thm:XTPi}.
To evaluate the algorithms, we define a \emph{True Estimation Rate} metric, $\textrm{TER}(x_t,\hat{x}_t)=1-(||x_t-\hat{x}_t||^2_2)/n$, which is the fraction of nodes whose states are correctly estimated.
Throughout (unless otherwise stated), the parameters are $\alpha=0.2,p=0.8,q=0.8,\rho_{ij}=0.1$, and we run each experiment for 20 time steps in 100 independent random trials.
We assume that the first node is the external threat that is known to be compromised and cannot be cleaned (we also exclude this node for evaluation purposes).
The initial state is then $x_0=(1,0,0,0,\dots)$; the initial beliefs of the algorithms are $P_0=(1,0.5,0.5,\dots)$ and $\Pi_0=\Pi^*(P_0)$ for MFA and BKF, respectively.
We chose this initial state and beliefs to reflect that the network begins clean (with the exception of source node 0) but the defenders assume they have no information about the network and choose to initialize it to the uniform random  distribution.
Using this choice, the experiments can demonstrate the speed of convergence of the algorithms from an uninformed initialization to accurate estimations.
Additionally, at each time step, we assume 2 nodes are randomly selected for cleaning (excluding node 1).

\begin{figure}[!htbp]
    \centering
\includegraphics[scale=0.1892]{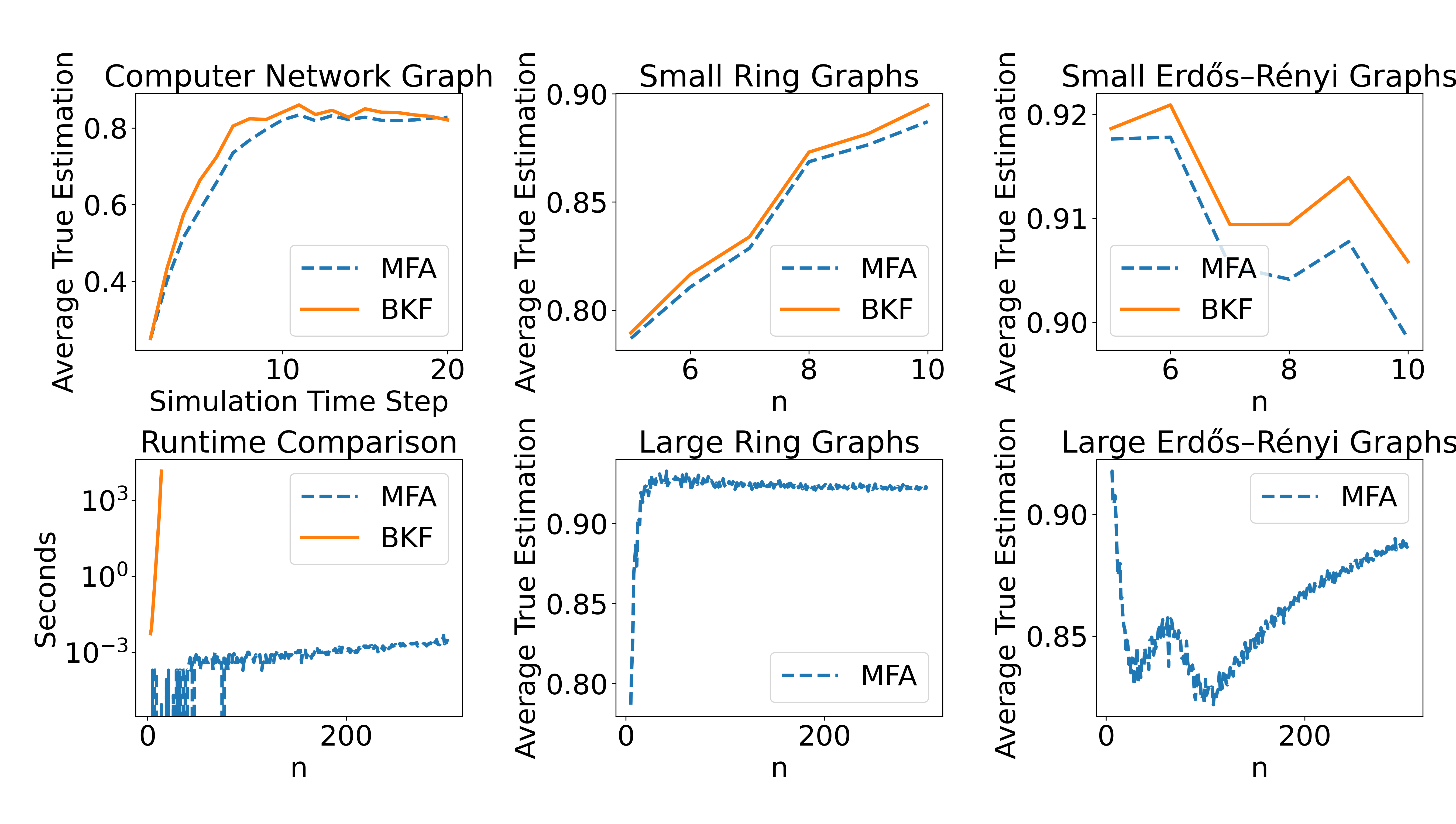}
\caption{Performance Comparison between the MFA and BKF Estimators. On the top-left we observe a comparative performance between the algorithms on the Computer Network Graph, and in the bottom left we compare runtimes on ring graphs.
In the center plots, we observe how the performance on ring graphs varies with the size of network. In the right plots, we report the same experiment but on Erd\H{o}s-Rényi graphs.}
    \label{fig:TD}
\end{figure}

To compare the performance between the MFA and BKF algorithms in terms of TER and runtime, we provide six experiments shown in Figure~\ref{fig:TD}.
In the top-left plot, we evaluate the True Estimation Rate at each timestep between both algorithms via an 11-node network proposed in \cite{kazeminajafabadi2024optimal}, which we term the \emph{Computer Network Graph}.
We take $p=0.2,q=1$ to match that work as closely as possible, noting that all of our nodes would be considered OR nodes (we did not implement their notion of AND nodes in this work) and that we include the external attacker as a node (making it an 11-node graph).
We observe that MFA follows BKF relatively closely at all time steps; given that the BKF algorithm is optimal in expectation, it is possible for the MFA algorithm to occasionally outperform BKF on specific sample paths.

The lower-left plot of Figure~\ref{fig:TD} is a run-time comparison between the two algorithms on ring graphs.
In the plot, $n$ denotes the number of nodes and the average runtime of 5 independent simulation runs is recorded for 20 time steps each.
As expected, the BKF algorithm is known to take exponential time (and we were able to compute up to $n=14$ for BKF), which is verified by the exponential curve in the plot.
It should be noted that for systems not much larger than the ones considered (e.g., a network with 30 nodes), the BKF estimator will no longer be practical, while the MFA estimator is.

The two middle plots present the changes in average performance on ring graphs of variable size, where TER is averaged across both 100 trials and all 20 time steps.
In the top-center plot, we evaluate $n\in \{5,\dots,10\}$ and observe that MFA and BKF perform very closely. In the bottom-center plot, we forgo computing BKF and plot MFA's performance on $n\in \{5,\dots,300\}$.
Interestingly, we observe a dramatic increase in performance as the number of nodes increases on small ring graphs, then leveling out around 0.93 for graphs of sufficient size.

In the two rightmost plots, we present the same study as the middle plots, but on Erd\H{o}s-Rényi (ER) graphs.
For these experiments, we define ER graphs where each edge exists with probability $0.2$, and we generate a new graph for each trial.
In the top plot, we observe that MFA and BKF follow each other closely;
in the bottom plot, MFA's performance is the strongest on small graphs (which may be easy to predict because they have
few edges).
Then, we observe a small increase in performance for certain moderately sized graphs, before the performance drops again.
Finally, as graphs grow large TER steadily improves again.
This suggests that certain edge densities may be more favorable for estimation than others, due to the fixed edge probability across all ER graph sizes.


\section{Conclusion}
We have presented a tractable algorithm for state estimation in a networked FlipIt model.
Our theoretical and simulation results show that our proposed method approximates the proven optimal BKF approach well. 
One interesting open problem is to prove bounds on the accuracy of this estimator.
\bibliographystyle{ieeetr}
\bibliography{references}

\online{
\section{Appendix}
\label{sec:proofs}
We now provide proofs for Theorems~\ref{thm:max entropy}, \ref{thm:XMPi}, \ref{thm:XTPi}.
Throughout, let $X_{r1} = \{x^i \in X \mid x^i_r = 1\}$ and $X_{r0} = \{x^i \in X | x_r = 0\}$.

\subsection*{Proof of Theorem~\ref{thm:max entropy}}
\begin{proof}
Let $P\in[0,1]^n$. We cast problem \eqref{eq:Pistar def}  as the following equivalent minimization problem:

\begin{align}
    &\min_{\Pi \in \mathbb{R}^{2^n}} \sum_i \Pi_i\ln \Pi_i \\
    &\textrm{ s.t.}\quad -\Pi_i \leq 0 \quad \forall i \label{eq:LM:nonnegative constraint} \\
    &\phantom{\textrm{ s.t.}\quad}-1+\sum \Pi_i=0 \label{eq:LM:probability constraint}\\
    &\phantom{\textrm{ s.t.}\quad}(X\Pi)_l-P_l=0\quad \forall l\in V. \label{eq:LM:P constraint}
\end{align}
We observe that the objective is the negative entropy function, which is known to be strongly convex, differentiable, and bounded from below. Moreover, the constraints constitute a collection of linear equalities and inequalities, and thus the feasible set is closed and convex. Therefore, this problem has a unique minimizer $\Pi^*$. 

We notice that Slater's constraint qualification is also satisfied in this instance. Hence, we can determine the unique minimizer by finding a solution to the KKT conditions. In addition to constraints \eqref{eq:LM:nonnegative constraint}-\eqref{eq:LM:P constraint}, the KKT conditions also consist of complementary slackness,
\begin{equation}
    \lambda_i \Pi_i = 0, \quad \forall i,
\end{equation}
dual feasibility $\lambda_i \geq 0$, and the zero-gradient condition $\nabla_{\Pi} L(\Pi,\lambda,\nu) = 0$, where the Lagrangian is written as
\begin{equation}
    \begin{aligned}
        L(\Pi,\lambda,\nu)&= \sum_i \Pi_i\ln \Pi_i - \sum_i \lambda_i \Pi_i \\
        &\quad+\nu_0(-1+\sum_{i} \Pi_i) + \sum_{l\in V}\nu_l(-P_l+\sum_i x_l^i\Pi_i).
    \end{aligned}
\end{equation}
Here, we use $\lambda_i$ as the multiplier associated with each inequality constraint of \eqref{eq:LM:nonnegative constraint}, $\nu_0$ associated with \eqref{eq:LM:probability constraint}, and $\nu_\ell$ associated with each constraint of \eqref{eq:LM:P constraint}. Each element of the gradient is given by
\begin{equation}
    \frac{\partial L}{\partial \Pi_i}=1+\ln\Pi_i - \lambda_i +\nu_0 +\sum_l \nu_lx^i_l.
\end{equation}
We set $\frac{\partial L}{\partial \Pi_i}=0$ and solve for $\Pi_i$, obtaining
\begin{equation}\label{eq:LM:Pi}
    \Pi_i= e^{\lambda_i} e^{-\nu_0-1}\prod_{l\in V}e^{-\nu_lx^i_l}.
\end{equation}
Here, we will suppose (which is standard) that $\lambda_i = 0$ for each $i$ and proceed with the analysis under this assumption. This choice automatically satisfies the complementary slackness and dual feasibility conditions, regardless of the choices of $\Pi_i$. Substituting into \eqref{eq:LM:probability constraint}, we obtain
\begin{equation}\label{eq:nu0}
    e^{-\nu_0-1}=\frac{1}{\sum_{x^i\in X}\prod_{l\in V} e^{-\nu_lx^i_l}}.
\end{equation}
In constraint \eqref{eq:LM:P constraint}, the term $(X\Pi)_r$ for any $r\in V$ can be expressed as
\begin{equation}
\begin{aligned}
    &(X\Pi)_r = e^{-\nu_0-1} \sum_{x^i\in X} x^i_r \prod_{l\in V}e^{-\nu_l x^i_l} \\
    &=\frac{\sum_{x^i\in X}x^i_r \prod_{l\in V}e^{-\nu_l x^i_l}}{\sum_{x^i\in X}\prod_{l\in V} e^{-\nu_lx^i_l}}=\\
    &\frac{e^{-\nu_r}\sum_{x^i\in X_{r1}} \prod_{l\in V\setminus r}e^{-\nu_l x^i_l}}{e^{-\nu_r}\sum_{x^i\in X_{r1}} \prod_{l\in V\setminus r}e^{-\nu_l x^i_l}+\sum_{x^i\in X_{r0}} \prod_{l\in V\setminus r}e^{-\nu_l x^i_l}}.\\
\end{aligned}
\end{equation}
Here, we notice that in the denominator of the last line above, the summation involving $X_{r1}$ is equivalent to the summation involving $X_{r0}$, because the sets $X_{r1}$ and $X_{r0}$ are both of size $2^{n-1}$ and the terms in each of the summations do not depend on $x_r$. Therefore, the last line above is equivalent to $\frac{e^{-\nu_r}}{e^{-\nu_r}+1}$. From constraint \eqref{eq:LM:P constraint}, this yields $\nu_r=-\ln\frac{P_r}{1-P_r}$ for all $r \in V$. The multiplier $\nu_0$ can then be determined from \eqref{eq:nu0}.

We take the previous two derivations and substitute them back into \eqref{eq:LM:Pi}, to obtain
\begin{equation}
\begin{aligned}
    &\Pi_i=\frac{\prod_{l\in V}e^{x^i_l\ln P_l-x^i_l\ln(1-P_l)} }{\sum_{x^j\in X}\prod_{l\in V}e^{x^j_l\ln P_l-x^j_l\ln(1-P_l)}}\\
    &=\frac{(\prod_{l\in V}e^{\ln (1-P_l)})\prod_{l\in V}e^{x^i_l\ln P_l-x^i_l\ln(1-P_l)} }{(\prod_{l\in V}e^{\ln (1-P_l)})\sum_{x^j\in X}\prod_{l\in V}e^{x^j_l\ln P_l-x^j_l\ln(1-P_l)}}\\
    &=\frac{\prod_{l\in V}e^{x^i_l\ln P_l+(1-x^i_l)\ln(1-P_l)} }{\sum_{x^j\in X}\prod_{l\in V}e^{x^j_l\ln P_l+(1-x^j_l)\ln(1-P_l)}}\\
    &=\frac{\prod_{l\in V}\bigg(x^i_l P_l+(1-x^i_l)(1-P_l)\bigg) }{\sum_{x^j\in X}\prod_{l\in V}\bigg(x^i_l P_l+(1-x^i_l)(1-P_l)\bigg)}\\
    &=\prod_{l\in V}\bigg(x^i_l P_l+(1-x^i_l)(1-P_l)\bigg),
\end{aligned}
\end{equation}
where the fourth equality follows by taking cases on $x^i_l\in\{0,1\}$ and the final equality follows as the denominator sums to $1$, intuitively because the sum is
over distribution $\Pi$.
More technically, by selecting any $r\in V$ and breaking the sums over $X_{r1},X_{r0}$, we can see
that $P_r,1-P_r$ can be factored out in each case.
If this process is done for all $r\in V$, the sum will be over a singleton set of an empty product.
We thus have $\Pi_i$, $\lambda_i$, $\nu_0$, and $\nu_\ell$ that together satisfy all KKT conditions. This concludes the proof.
\end{proof}

To prove Theorem~\ref{thm:XMPi}, we need the following Lemma.
\begin{lemma}\label{thm:factoring lemma}
    If $\Pi=\Pi^*(P)$, then for any $\sigma\subseteq V$ we have
    \begin{equation}
        \prod_{l\in \sigma} P_{t,l}=\sum_{i\in \{1,\dots, 2^n\}} \Pi_i\prod_{l\in \sigma}x^i_l.
    \end{equation}
\end{lemma}
\begin{proof}
Let $\Pi=\Pi^*(P)$ for some $P\in[0,1]^n$, 
$\sigma\subseteq V$, $\sigma'$ be the complement set of $\sigma$, $X^\sigma=\{i\in \{1,\dots 2^n\}\mid x^i_l=1 \mbox{ }\forall l\in \sigma\}$, $k\notin \sigma$, and
$X^\sigma_{k1}=\{i\in X^\sigma\mid x^i_k=1\}$ and similarly for $X^\sigma_{k0}$.
Beginning with the right side of the desired equality,
    \begin{align*}
    &\sum_{i\in \{1,\dots, 2^n\}} \Pi_i\prod_{l\in \sigma}x^i_l=\sum_{i\in X^\sigma} \Pi_i\nonumber\\
    &=\prod_{l\in \sigma}P_{t,l}\sum_{i\in X^\sigma} \prod_{l\in \sigma'}(P_{t,l}x^i_l+(1-P_{t,l})(1-x^i_l))\nonumber\\
    &=\prod_{l\in \sigma}P_{t,l}\bigg(P_{t,k}\sum_{i\in X^\sigma_{k1}} \prod_{l\in \sigma'\setminus k}(P_{t,l}x^i_l+(1-P_{t,l})(1-x^i_l))\nonumber\\
    &+(1-P_{t,k})\sum_{i\in X^\sigma_{k0}} \prod_{l\in \sigma'\setminus k}(P_{t,l}x^i_l+(1-P_{t,l})(1-x^i_l))\bigg)\nonumber\\
    &=\prod_{l\in \sigma}P_{t,l}\sum_{i\in X^\sigma_{k1}} \prod_{l\in \sigma'\setminus k}(P_{t,l}x^i_l+(1-P_{t,l})(1-x^i_l))\nonumber\\
    &=\prod_{l\in \sigma}P_{t,l}.\nonumber
    \end{align*}
Note that the second equality follows as $i\in X^\sigma$ implies that $\prod_{l\in \sigma}P_{t,l}$ must factor from $\pi_i$ by definition.
The third equality follows by separating the sums by $X^\sigma_{k1}$ and $X^\sigma_{k0}$ and by factoring out the appropriate $P_{t,k}$ or $1-P_{t,k}$. 
The fourth equality then follows as each product is no longer a function of $x^i_k$, and the size of the sum can be halved.
Since this process was done generically for any $k\in \sigma'$, the final equality follows by an iterative argument by doing this procedure for each $k\in \sigma'$.
After all members of $\sigma'$ are factored out, they will be over a singleton set and the product will be over the empty set, concluding the proof.
\end{proof}

\subsection*{Proof of Theorem~\ref{thm:XMPi}}
\begin{proof}
Throughout the proof we drop time index subscripts and use subscripts to index elements of vectors instead.
Let $\Pi=\Pi^*(P)$ for some $P\in[0,1]^n$.
    We begin by showing 
    \begin{align*}
        &\Pr(x_{t+1,r}=1\mid \Pi^*(P),a)=\sum_{x^i\in X}x^i_r\sum_{x^j\in X}\Pi_jM_{ij}\\
        &=\sum_{x^j\in X}\Pi_j\sum_{x^i\in X}x^i_r\prod^n_{l=1}\bigg(\eta^j_l x^i_l+(1-\eta^j_l)(1-x^i_l) \bigg)\\
        &=\sum_{x^j\in X}\Pi_j\sum_{x^i\in X_{r1}}\eta^j_r\prod^n_{l=1,l\neq r}\bigg(\eta^j_l x^i_l+(1-\eta^j_l)(1-x^i_l) \bigg) \\
        &=\sum_{x^j\in X}\Pi_j\eta^j_r,
    \end{align*}
where the last equality follows an identical factoring argument as in the proof of Lemma~\ref{thm:factoring lemma}.
Letting $\tilde{a}_r=1+a_r(\alpha-1)$, we have
\begin{align*}
    &\sum_{x^j\in X}\Pi_j\eta^j_r\\
    &=\tilde{a}_r\sum_{x^j\in X}\Pi_j\bigg(1-\prod_{l\in D_r}(1-\rho_{lr}x^j_l)+x^j_r\prod_{l\in D_r}(1-\rho_{lr}x^j_l)\bigg)\\
    &=\tilde{a}_r\bigg(1-\sum_{x^j\in X}\Pi_j\prod_{l\in D_r}(1-\rho_{lr}x^j_l)\\
    &\qquad\qquad+\sum_{x^j\in X}\Pi_jx^j_r\prod_{l\in D_r}(1-\rho_{lr}x^j_l)\bigg).
\end{align*}
Now we consider each of the two remaining sum terms separately.
Let $\mathcal{P}(D_r)$ denote the power set of in-neighbors of $r$ and $\mathcal{P}_k(D_r)=\{\sigma\in \mathcal{P}(D_r)\mid |\sigma|=k\}$ be all of the elements of the power set of size $k$.
Using this, we have
\begin{align*}
    &\sum_{x^j\in X}\Pi_j\prod_{l\in D_r}(1-\rho_{lr}x^j_l)\\
    &=\sum_{x^j\in X}\Pi_j\bigg(1-\sum_{\sigma \in \mathcal{P}_1(D_r)}\prod_{l\in \sigma}\rho_{lr}x^j_l \\
    &\quad +\sum_{\sigma \in \mathcal{P}_2(D_r)}\prod_{l\in \sigma}\rho_{lr}x^j_l-\cdots + \sum_{\sigma \in \mathcal{P}_{|D_r|}(D_r)}\prod_{l\in \sigma}\rho_{lr}x^j_l \bigg)\\
    &=1-\sum_{\sigma \in \mathcal{P}_1(D_r)}\big(\prod_{l\in \sigma}\rho_{lr}\big)\sum_{x^j\in X}\Pi_j\prod_{l\in \sigma}x^j_l \\
    &\quad +\cdots + \sum_{\sigma \in \mathcal{P}_{|D_r|}(D_r)}\big(\prod_{l\in \sigma}\rho_{lr}\big)\sum_{x^j\in X}\Pi_j\prod_{l\in \sigma}x^j_l \\
    &=1-\sum_{\sigma \in \mathcal{P}_1(D_r)}\big(\prod_{l\in \sigma}\rho_{lr}\big)\prod_{l\in \sigma}P_l\\
    &\quad +\cdots + \sum_{\sigma \in \mathcal{P}_{|D_r|}(D_r)}\big(\prod_{l\in \sigma}\rho_{lr}\big)\prod_{l\in \sigma}P_l =\prod_{l\in D_r} (1-\rho_{lr} P_l),
\end{align*}
where the third equality is an application of Lemma~\ref{thm:factoring lemma}.
The next term follows a similar argument to the first three equalities, obtaining
\begin{align*}
    &\sum_{x^j\in X}\Pi_jx^j_r\prod_{l\in D_r}(1-\rho_{lr}x^j_l)\nonumber\\
    &=\sum_{x^j\in X}\Pi_jx^j_r-\sum_{\sigma \in \mathcal{P}_1(D_r)}\big(\prod_{l\in \sigma}\rho_{lr}\big)\sum_{x^j\in X}\Pi_jx^j_r\prod_{l\in \sigma}x^j_l \nonumber\\
    &\quad +\cdots + \sum_{\sigma \in \mathcal{P}_{|D_r|}(D_r)}\big(\prod_{l\in \sigma}\rho_{lr}\big)\sum_{x^j\in X}\Pi_jx^j_r\prod_{l\in \sigma}x^j_l \nonumber\\
    &=P_r-\sum_{\sigma \in \mathcal{P}_1(D_r)}\big(\prod_{l\in \sigma}\rho_{lr}\big)\prod_{l\in \sigma\cup r}P_l \nonumber\\
    &\quad +\cdots + \sum_{\sigma \in \mathcal{P}_{|D_r|}(D_r)}\big(\prod_{l\in \sigma}\rho_{lr}\big)\prod_{l\in \sigma\cup r}P_l\nonumber\\ 
    &=P_r\prod_{l\in D_r} (1-\rho_{lr} P_l) \nonumber\\
\end{align*}
where the second equality again follows from Lemma~\ref{thm:factoring lemma}.
Using the above two derivations as desired, we obtain 
\begin{equation*}
    (XMT)_r=\tilde{a}_r\bigg( P_r+(1-P_r)
    \bigg[1-\prod_{l\in D_r} 
    (1-\rho_{lr} P_l)\bigg]\bigg).
\end{equation*}
\end{proof}

\subsection*{Proof of Theorem~\ref{thm:XTPi}}
\begin{proof}
Throughout the proof, we drop the time index subscripts and use subscripts to index elements of vectors instead.
Let $\Pi=\Pi^*(P')$ for some $P'\in[0,1]^n$ and we start by considering $\frac{(X(T\circ \Pi))_l}{T^\top\Pi}$,
where $M,T$ are constructed appropriately via \eqref{eq:def Mt} and \eqref{eq:def Tt}, respectively.
Beginning with $T^\top\Pi$, we obtain
    \begin{align*}
        &T^\top\Pi\\
        &=\sum_{x^i\in X} \prod_{r\in V} \big(P'_r x^i_r+(1-P'_r)(1-x^i_r)\big) \Pr(y_r\mid x_t=x^i)\\
        &=\sum_{x^i\in X} \prod_{r\in V} P'_rx^i_r\bigg(py_r+\tilde{p}(1-y_r)\bigg)\\
        &\qquad+(1-P'_r)(1-x_r)\bigg(\tilde{q}y_r+q(1-y_r)\bigg)\\
        &=\sum_{x^i\in X} \prod_{r\in V} F_r(y,x^i)\\
        &=pP'_l y_l+\tilde{p}P'_l(1-y_l)\sum_{x^i\in X_{k1}} \prod_{r\in V\setminus l} F_r(y,a^i)\\
        &+(1-P'_l)(\tilde{q}y_l+q(1-y_l))\sum_{x^i\in X_{k0}} \prod_{r\in V\setminus l} F_r(y,x^i)\\
        &=\bigg(y_l\big[pP'_l+\tilde{q}(1-P'_l)\big]\\
        &+(1-y_l)\big[\tilde{p}P'_l+q(1-P'_l)\big]
        \bigg)\sum_{x^i\in X_{k1}} \prod_{r\in V\setminus l} F_r(y,x^i)
    \end{align*}
letting $F_r(y,x^i)$ be the appropriate function.
The final equality follows because for each $x^i\in X_{l1}$ there is a $x^j\in X_{l0}$ such that $x_r^i=x_r^j$ for all $r\in V\setminus l$.
Next, we evaluate the numerator.
    \begin{align*}
        &(X(T\circ \Pi))_l=\sum_{x^i\in X}x^i_l\prod_{r\in V} F_r(y,x^i)\\
        &=\bigg(y_lp P'_l +(1-y_l)\tilde{p}P_l\bigg)\sum_{x^i\in X_{l1}}\prod_{r\in V\setminus l} F_r(y,x^i).
    \end{align*}
Using these two derivations to reconstruct the original fraction, we see that $\sum_{x^i\in X_{r1}}\prod_{r\in V\setminus l} F_r(y,x^i)$ appears in both the numerator and the denominator (and is nonzero), leading to the desired equality and completing the proof.
\end{proof}
}

\end{document}